\documentclass[a4paper,UKenglish,cleveref, autoref, thm-restate]{lipics-v2021}
\usepackage[utf8]{inputenc}
\usepackage{graphicx}
\usepackage{xcolor}
\usepackage{mathtools}
\usepackage{xspace}
\usepackage{amsthm}
\usepackage{caption}
\usepackage{subcaption}
\usepackage{url}
\usepackage{amsfonts} 
\usepackage{hyperref}

\newcommand*{\gpos}{$G_{pos}($POS CNF$)$\xspace}
\newcommand*{\T}{Trudy\xspace}
\newcommand*{\F}{Fred\xspace} %{Falco\xspace}
\newcommand*{\DB}{Dots \& Boxes\xspace}

\newcommand*{\true}{\textsc{true}\xspace}
\newcommand*{\false}{\textsc{false}\xspace}

\newcommand*{\org}{\textsf{or}\xspace}

\newcommand*{\crossover}{\textsf{crossover}\xspace}
\newcommand*{\wire}{\textsf{wire}\xspace}
\newcommand*{\variable}{\textsf{variable}\xspace}
\newcommand*{\clause}{\textsf{clause}\xspace}
\newcommand*{\cmax}{\ensuremath{\mathcal{C}_{\max}}\xspace}

\title{Dots \& Boxes is PSPACE-complete}

\author{Kevin Buchin}{Department of Mathematics and Computer Science, TU Eindhoven, Netherlands}{k.a.buchin@tue.nl}{https://orcid.org/0000-0002-3022-7877}{}

\author{Mart Hagedoorn}{Department of Mathematics and Computer Science, TU Eindhoven, Netherlands}{m.h.hagedoorn@student.tue.nl}{}{}

\author{Irina Kostitsyna}{Department of Mathematics and Computer Science, TU Eindhoven, Netherlands}{i.kostitsyna@tue.nl}{https://orcid.org/0000-0003-0544-2257}{}

\author{Max van Mulken}{Department of Mathematics and Computer Science, TU Eindhoven, Netherlands}{m.j.m.v.mulken@student.tue.nl}{}{}

\authorrunning{K. Buchin, M. Hagedoorn, I. Kostitsyna, M. v. Mulken }

\Copyright{Kevin Buchin, Mart Hagedoorn, Irina Kostitsyna, Max van Mulken} 

 %mandatory: Please choose ACM 2012 classifications from https://dl.acm.org/ccs/ccs_flat.cfm 


\category{} %optional, e.g. invited paper

\relatedversion{} %optional, e.g. full version hosted on arXiv, HAL, or other respository/website
\nolinenumbers %uncomment to disable line numbering

\hideLIPIcs  %uncomment to remove references to LIPIcs series (logo, DOI, ...), e.g. when preparing a pre-final version to be uploaded to arXiv or another public repository

\EventEditors{John Q. Open and Joan R. Access}
\EventNoEds{2}
\EventLongTitle{42nd Conference on Very Important Topics (CVIT 2016)}
\EventShortTitle{CVIT 2016}
\EventAcronym{CVIT}
\EventYear{2016}
\EventDate{December 24--27, 2016}
\EventLocation{Little Whinging, United Kingdom}
\EventLogo{}
\SeriesVolume{42}
\ArticleNo{23}
\begin{document}

\maketitle

\begin{abstract}
  Exactly 20 years ago at MFCS, Demaine posed the open problem whether the game of \DB is PSPACE-complete. \DB has been studied extensively, with for instance a chapter in Berlekamp et al.~\emph{Winning Ways for Your Mathematical Plays}, a whole book on the game~\emph{The Dots and Boxes Game: Sophisticated Child's Play} by Berlekamp, and numerous articles in the \emph{Games of No Chance} series. While known to be NP-hard, the question of its complexity remained open. We resolve this question, proving that the game is PSPACE-complete by a reduction from a game played on propositional formulas.
\end{abstract}

\section{Introduction}

\emph{\DB} is a popular paper-and-pencil game that is played by two players on a grid of dots.
The players take turns connecting two adjacent dots.
If a player completes the fourth side of a unit box, the player is awarded a point and an additional turn.
When no more moves can be made, the player with the highest score wins the game.\footnote{For a visual explanation of the game see \url{https://youtu.be/KboGyIilP6k}, last accessed 6.5.2021}

Originally described in 1883~\cite{lucas1883recreations}, \DB has since received a considerable amount of attention in the research community.
In \textit{Winning Ways for Your Mathematical Plays}, Berlekamp, Conway, and Guy~\cite{winning_ways} were among the first to discuss a number of interesting mathematical properties of the game. 
Later, Berlekamp~\cite{berlekamp2000dots} wrote an entire book \emph{The Dots-and-Boxes game: Sophisticated Child's Play} about the game, in particular
discussing winning strategies in particular positions.
Since then, the mathematics of \DB and variants has been discussed in many papers and books~\cite{AICHHOLZER200542,albert2019lessons,berlekamp2000forcing,collette2015narrow,demainehearnplayagain,gpc,horiyama2017sankaku,johnson2014combinatorial,MeRo1988,Now1991,siegel2013combinatorial,Wes1996a}. There is also a rich body of work on solvers for \DB~\cite{barker2011solving,barker2012solving,buzzard2014playing,knittel2007concept,zhuang2015improving}.

Berlekamp et al.~\cite{winning_ways} argue that deciding the winner of a generalized version of \DB, called \emph{Strings-and-Coins}, is NP-hard. In this game, players take turns in removing edges of a given graph, scoring a point when they isolate a vertex.
When restricted to the dual graph of a square grid, this corresponds to a dual formulation of \DB. 
Eppstein~\cite{Epp} notes that the reduction given by Berlekamp et al.\ should extend to \DB, and a formal proof of the NP-hardness is given in~\cite{buchin_et_al:LIPIcs:2020:12237}.

Exactly 20 years ago at MFCS, Demaine posed the open problem whether \DB is PSPACE-complete~\cite{demaine2001playing}.
Bounded two-player games, like \DB, (that is, games in which the number of moves is bounded) naturally lie in PSPACE,
since a Turing machine using space polynomial in the board size is able to search the entirety of the game space. Often, these games are also PSPACE-hard~\cite{demaine2001playing}.
PSPACE-hardness of many bounded two-player games is shown by a reduction from \emph{Generalized Geography}, which is proven PSPACE-complete by Lichtenstein and Sipser~\cite{lichtenstein1980go}.
For example, the PSPACE-completeness of Reversi~\cite{reversi}, uncooperative UNO~\cite{uno}, and Tic-Tac-Toe~\cite{hsieh2007fairness} were shown by a reduction from Generalized Geography.
However, unlike \DB, the setting of Generalized Geography prescribes a stricter order on players' moves, making a reduction to \DB challenging to obtain.

In their seminal work, Hearn and Demaine~\cite{bhearn,gpc} introduce \emph{Constraint Logic}, a framework for analyzing complexity of games and puzzles.
Inspired by Flake and Baum's proof of Rush Hour~\cite{FLAKE2002895}, it specifies a type of game played on a constraint graph. The framework includes bounded/unbounded state spaces and single/two-player variations.
In the same work, Hearn and Demaine go on to provide a number of simpler reductions for various known PSPACE-complete games (including Rush Hour), as well as new proofs for several PSPACE-complete games.
However, the Constraint Logic framework is intended for proving hardness of partisan games (games in which the moves available to the two players are different), whereas \DB is not a partisan game.

Strings-and-Coins and the related game of \emph{Nimstring} were very recently (while we were preparing this submission) proven to be PSPACE-complete by Demaine and Diomidov~\cite{nimstring_pspace} by a reduction from a game on a DNF formula $G_{pos}($POS DNF$)$~\cite{schaefer}.
But, as they point out, their results do not apply to \DB, since the game positions they construct rely on multi-graphs (which additionally are neither planar nor have a maximum degree of 4). Specifically, they propagate signals through multi-edges consisting of a polynomial number of parallel edges, and the winner is the player who removes the last edge. As consequence, our reduction bears little commonalities with theirs.

In this paper, we prove that \DB is PSPACE-complete by a reduction from \gpos. The starting point of our construction are strategies for \DB endgames that were also used to prove NP-hardness.
However, the NP-hardness is proven by having one player be in control, and there being only one way for the other player to respond.
This de facto makes the game to be 1-player game.
For PSPACE-hardness we need both players to have choices, making it a true 2-player game.
This gives a lot of freedom to the players, and makes it much more difficult to construct gadgets to control the gameplay, in particular because moves and scoring opportunities for one player---if not played immediately---are also available to the other player. 

In Section~\ref{subsec:DB-gameplay} we discuss the gameplay of \DB in detail, and introduce terminology coined by Berlekamp et al.~\cite{winning_ways}.
In Section~\ref{sec:pspace} we present the general structure of our reduction, and then describe our gadgets in Section~\ref{sec:gadgets}. In  Section~\ref{sec:construction} we first show that the players' strategies, which we intend the players to use, are optimal for them and finally prove PSPACE-hardness. 

\subsection{\DB}\label{subsec:DB-gameplay}

On the surface, \DB is quite a simple game. 
The starting and a typical final position for a $10\times 10$ grid are shown in Figure~\ref{fig:typical-position}.
We refer to the players playing the blue and the red colors as \T and \F, respectively.
The color of a line connecting two dots indicates which player drew it, and the color of a box---which player closed it.

Consider a dual graph $G$ of a board of \DB, where a node in $G$ corresponds to a box or the unbounded face, and a pair of nodes in $G$ is connected with an edge if the corresponding move is still available, i.e., the line between the boxes has not been drawn. 
Let the \emph{degree} of a box be the degree of the corresponding node in $G$.

\begin{figure}[t]
     \begin{subfigure}[t]{0.31\textwidth}
         \centering
        \includegraphics[page=32]{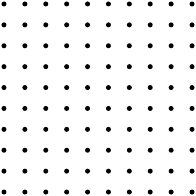}
        \subcaption{Starting position of a \DB game with a $10 \times 10$ grid.}
    \end{subfigure}
    \hfill
    \begin{subfigure}[t]{0.31\textwidth}
        \centering
        \includegraphics[page=31]{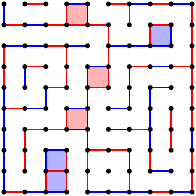}
        \caption{An example of a loony endgame state of the game.}
        \label{fig:loony-endgame}
    \end{subfigure}
    \hfill
     \begin{subfigure}[t]{0.31\textwidth}
         \centering
        \includegraphics[page=33]{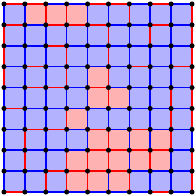}
        \subcaption{A possible end position. \T (blue) won with 62 boxes versus \F's (red) 19 boxes.}
        \label{fig:typical-end}
    \end{subfigure}
    \caption{Typical starting, intermediate, and final position of a \DB game.}
    \label{fig:typical-position}
\end{figure}

In \DB, a typical game usually results in a board state that consists exclusively of moves that open the possibility for the opponent to claim a number of boxes in their next turn (see Figure~\ref{fig:loony-endgame}).
That is, in this state there are no degree-1 boxes, but any move made by a player creates a degree-1 box that can be immediately claimed by the opponent.
Consider such a board configuration $S$ and any available move $\ell$ in it.
At least one box $b$ incident to $\ell$ has degree two in $S$ (before the move $\ell$ is made).
Consider a maximal component $\sigma$ of degree-2 boxes in $S$ containing $b$.
There are two cases, either $\sigma$ is a \emph{chain} ending in boxes of degree higher than two (or the outer face), or $\sigma$ is a \emph{cycle}.
Then we say that a player making the move $\ell$ \emph{opens} the chain (cycle) $\sigma$ for the opponent.

\begin{figure}[b!]
    \centering
    \includegraphics[page=39]{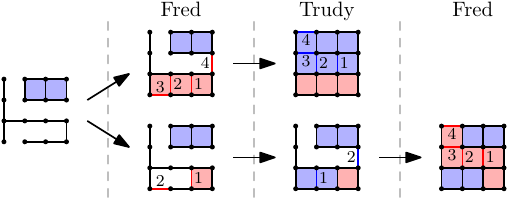}
    \caption{Two possible plays that \F (red) can do. \F can choose to claim all the available boxes (top) and lose the game, or to perform a double-dealing move sacrificing two boxes (bottom, second state, edge 2), and win the game. The order of the edges that are played by \T or \F in one turn is indicated by edge labels. This example is borrowed from \textit{Winning Ways, chapter 16}~\cite{winning_ways}.}
    \label{fig:double-dealing-better}
\end{figure}

To devise a good strategy for \DB, it is important to note that a player is not obliged to claim a box whenever they have the ability to do so.
While seemingly counter-intuitive, it is sometimes beneficial for a player to sacrifice a small number of boxes for long-term gain.
Consider the position in Figure \ref{fig:double-dealing-better}, and let it be \F's (red) turn.
Here, it may seem intuitive for \F to claim the bottom three boxes (Figure~\ref{fig:double-dealing-better} (top)).
However, after doing so \F has to make an extra move, allowing \T (blue) to claim the remaining four boxes and win the game.
On the other hand, by sacrificing two boxes (Figure~\ref{fig:double-dealing-better} (bottom)), \F can force \T to make another move and open the middle chain for him to claim.
That way, \F loses two boxes in the bottom chain, but gains all four boxes in the middle chain, securing the win.

In \textit{Winning Ways}, Berlekamp et al.~\cite{winning_ways} refer to the moves sacrificing a small number of boxes but passing the turn onto the opponent as \emph{double-dealing} moves.
Double-dealing moves can be made in chains of boxes, sacrificing two boxes, and in cycles, sacrificing four boxes (see Figure~\ref{fig:double-dealing}).
Each double-dealing move is usually immediately followed by the opponent making at least one \emph{double-cross} move, i.e., a move that closes two boxes at once.
These double-dealing and double-cross moves are essential for players that want to consistently win games of \DB, and will be used in the reduction later.

\begin{figure}[t]
    \centering
    \includegraphics[page=30]{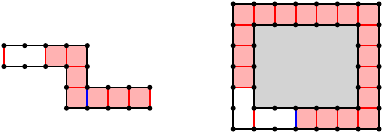}
    \caption{Double-dealing move by \F (red). If \T (blue) opens a chain (or a cycle), \F can claim a sequence of boxes. To pass the turn back to \T, \F can leave two (or four) boxes unclaimed.}
    \label{fig:double-dealing}
\end{figure}

Note that double-dealing moves are only possible in \emph{long} chains of at least three boxes, and in cycles.
(Chains of length one do not have enough boxes for a double-dealing move, and a chain of length two can be opened by selecting the middle edge, thus preventing the opponent from playing a double-dealing move.)
Thus, opening a long chain or a cycle, if there are other moves available, is often a bad idea.
Berlekamp et al.~\cite{winning_ways} refer to such moves as \emph{loony} moves.

Making loony moves is not always a choice.
If, at some point in the game, all unclaimed boxes are part of long chains and cycles, the only possible moves are loony moves (Figure~\ref{fig:loony-endgame}).
Such positions are referred to as a \emph{loony endgames}.
Note that in chains of length $\geq 4$ and cycles of length $\geq 8$, the player making the double-dealing moves scores at least as many boxes as their opponent.
Thus, in loony endgames with chains of length $\geq 4$ and cycles of length $\geq 8$, under optimal play, the game consists of one player making loony moves (opening chains and cycles), and the other player claiming all but two or four boxes, and making double-dealing moves to pass the turn back to the opponent~\cite{winning_ways}.
Here, the player making the double-dealing moves is always better off, since each chain or cycle yields at least as many boxes to this player as it yields to their opponent.
This player is thus referred to as being \emph{in control} of the game.
The benefit of being in control can be seen in  Figure~\ref{fig:typical-end}, which is the end result of \T being in control of the loony endgame shown in Figure~\ref{fig:loony-endgame}.

In \textit{Winning Ways}, Berlekamp et al.~\cite{winning_ways} state that finding a winning strategy in the loony endgame for the player who is not in control is NP-hard.
They argue that maximizing the number of disjoint cycles will maximize the score of the player not in control, since double-dealing moves in cycles yield twice as many boxes as double-dealing moves in chains.
Since this property is important for our reduction, we restate it here and, for completeness, present the argument in the appendix.

% [Lemma: show that maximizing the number of cycles yields the highest score for the player that is not in control once in the loony endgame]

\begin{restatable}{lemma}{lemmaxcycles}\label{lem:max_cycles}
Let the configuration of a loony endgame contain $k$ boxes with degree higher than $2$, let $T$ be the sum of the degrees of these boxes, and let $c$ be the maximum number of disjoint cycles in the configuration.
Then, the player who is \emph{not} in control can claim at most $4c+T-2k-4$ boxes. %maximizes their score when the number of disjoint cycles is maximized.
\end{restatable}

\section{Structure of the construction}\label{sec:pspace}

% [Define \gpos]
To show that \DB is PSPACE-hard we reduce from the game \gpos, introduced and proven PSPACE-complete by Schaefer~\cite{schaefer}.
The game is played by two players, \T and \F, on a positive CNF formula $\mathcal{F}$.
The players take turns picking a variable that has not yet been chosen. Variables picked by \T are set to \true, variables picked by \F are set to \false. When all variables have been chosen, the game ends. \T wins if formula $\mathcal{F}$ evaluates to \true, and \F wins if formula $\mathcal{F}$ evaluates to \false. 

Given a positive CNF formula $\mathcal{F}$ with $n$ variables and $m$ clauses, we construct an instance of \DB in which \T has a winning strategy if and only if she also has a winning strategy in the corresponding instance of \gpos.
For simplicity we assume that $n$ is even, so that \T and \F get to assign values to the same number of variables.
If the number of variables in $\mathcal{F}$ is odd, we can introduce dummy variables without changing the outcome of a game such that the total number of the variables becomes even.
For each variable and clause of $\mathcal{F}$ we construct a \variable and a \clause gadget, respectively.
We place the \variable gadgets in a row at the top of the board of \DB, and the \clause gadgets in a row at the bottom.
We connect the \variable gadgets to their corresponding \clause gadgets using the \wire gadgets, which transfer the values of the variables to the clauses.
If a clause consists of more than one variable, the wires from these variables must pass through an \org gadget.
Since the signals propagating from the variables may need to cross each other, we construct a \crossover gadget that preserves the values in the two crossing wires.
In our instance of \DB, only the gadgets contain available moves.
The remaining boxes on the board have all the incident edges present.

As we detail in Section~\ref{sec:construction}, after the values of the variables are set, the game enters a loony endgame where \F is in control.
Then \T's winning strategy reduces to selecting a maximum set of disjoint cycles \cmax in the remaining configuration (Lemma~\ref{lem:max_cycles}).
To maximize her score, \T opens all the chains outside of \cmax first, gaining two boxes per chain, and opens the chosen cycles last, gaining four boxes per cycle.
The optimal play for \F is to ensure that he will be in control when the loony endgame starts.
After entering the loony endgame, simply making double-dealing moves until his very last turn is optimal for \F.

\subparagraph{Signal representation}
%In the \DB instance $\delta$ signal values have to be passed down all the way from the \variable gadgets down to the \clause gadgets. Before the signal values are chosen, the gadgets will consist of partially overlapping cycles. 
Most of our gadgets consist of partially overlapping cycles of boxes.
%After the players alternatingly set all the variables to \true or to \false, Trudy's winning strategy reduces to selecting a maximal set of disjoint cycles in the remaining configuration.
The choice of a set of disjoint cycles determines the value of a signal.
For example, in Figure~\ref{fig:signals} the choice of the left vs. right cycle can encode the value \true vs. \false.
Of course, \T could join the cycles together to select the outermost cycle, but this, as we show later, will not be more beneficial. 

\begin{figure}[t]
     \centering
     \begin{subfigure}[b]{0.4\textwidth}
         \centering
         \includegraphics[page=2]{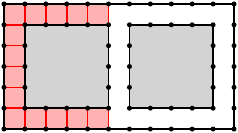}
         \label{fig:signal_2}
     \end{subfigure}
     \begin{subfigure}[b]{0.4\textwidth}
         \centering
         \includegraphics[page=3]{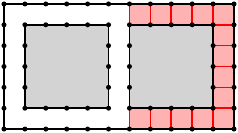}
         \label{fig:signal_1}
     \end{subfigure}
        \caption{A choice of a cycle can encode the value of a signal.}
        \label{fig:signals}
\end{figure}

% [Properties of the gadgets:]
% [Since red is not in control, they are responsible for signal propagation] ->
% [setting the variables should be done with the last remaining non-loony moves, so that the loony endgame is entered once all variables are set]
\subparagraph{Variable assignment}
As both players must have a choice in picking which variable to set, the instance of \DB cannot yet be in a loony endgame.
Thus, the \variable gadgets, which we describe in detail in Section~\ref{sec:variable}, contain non-loony moves instrumental in setting the value of a variable.
%
%For every variable $x$ in $\mathcal{F}$ a \variable gadget will be constructed. The construction of these variables will be discussed in Section~\ref{sec:gadgets}. 
%
%To properly model the reduction from \gpos, 
We ensure that the optimal behavior of both players results in the variables being set in alternating fashion, where \T sets them to \true, and \F sets them to \false.
%After a variable has been appointed a value, it propagates to the corresponding \clause gadgets.
%We will ensure that, 
Once all variables are set, the loony endgame is entered.
%It is easier to predict the final score of the player who is \textit{not} in control of the loony endgame than it is to predict the final score of the player who is in control, since the player not in control will only gain boxes in increments of 2 (double-dealings in chains) and 4 (double-dealings in cycles). 
At this point \F is in control of the game, and it is up to \T to maximize her score by maximizing the number of disjoint cycles in \cmax.
The optimal play by \T results in a correct propagation of the signals from the variables to the clauses.
%Therefore, we will ensure that, once all variables are set, \F will always be in control of the resulting loony endgame. This means it will be up to \T to propagate the signal in such a way that is optimal; maximizing the number of disjoint cycles (Lemma~\ref{lem:max_cycles}).

% [variables can output multiple wires (as many as required)]

% [or gadget acts as an or gate and gives an equal amount of boxes in all valid situations]
\subparagraph{Remaining constraints and scoring}
To ensure that optimal play by both players in the instance of \DB corresponds to a valid \gpos game, our gadgets need to give a specific number of boxes to \T depending on the signal values.
We will show that after the variable values have been set, under optimal play, \T can maximize her score only if the signals are propagated correctly.
Every gadget, except for the \clause, yields the same number of disjoint cycles independent of the values of the signals passing through the gadget.
Only the \clause gadget gives more cycles to \T if a \true signal reaches it.
Exactly half of the variables are set to \true, and half to \false.
Thus we can tune the starting score count between \T and \F such that the game is won by \T if and only if all the clauses are satisfied.

\section{Gadgets} \label{sec:gadgets}
% [give each gadget, explain why they satisfy the constraints above]

%Our reduction from \gpos requires four different gadgets (variable, wire, \org and clause), which we will provide here.
In this section we provide the details of the gadgets used in our reduction.
When describing the gadgets below, for a simpler exposition, we assume that the moves that \T and \F make follow the following sequence.
First, in the first $n$ moves \T and \F set all the variables to \true and \false respectively.
Afterwards, when the loony endgame is entered, the order in which \T selects which cycles to add to the disjoint set of cycles \cmax is from the top to bottom, that is, from the variables, through the outgoing wires, through the crossover and or gadgets, and finally down to the \clause gadgets.
Later, in Lemma~\ref{lem:strategies}, we will show that, indeed, under optimal gameplay \T and \F start by setting all the variables.
Furthermore, we will argue that the outcome of the game depends only on the choice of the cycles in \cmax, and not on the order in which \T selects them.

\subsection{Basic wiring}\label{sec:wires}

\begin{figure}[b]
     \centering
%     \begin{subfigure}[b]{0.5\textwidth}
%         \centering
         \includegraphics[page=8]{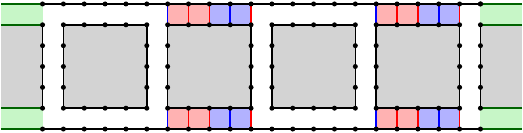}
%     \end{subfigure}
\\[0.4cm]
%     \vspace{0.5cm}
%     \begin{subfigure}[b]{0.5\textwidth}
%         \centering
         \includegraphics[page=7]{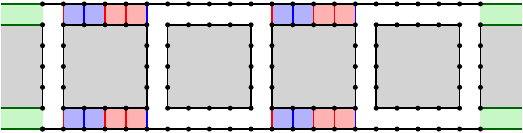}
%     \end{subfigure}
        \caption{A \wire gadget consisting of four overlapping cycles and two ways of selecting disjoint cycles. Shown in green are the connections to the adjacent gadgets. Selecting odd cycles in \cmax corresponds to \true (top), and selecting even cycles corresponds to \false (bottom).}
        \label{fig:wire-setting}
\end{figure}

Signals from the \variable gadgets are propagated to the \clause gadgets through wires.
A wire consists of a chain of an even number of partially overlapping cycles (see 
Figure~\ref{fig:wire-setting}).
The first cycle in the wire overlaps with the gadget from which the signal is propagated, and the last cycle overlaps with the gadget towards which the signal is propagated.
Consider some wire $w$, let $C_i$ be its first cycle overlapping with gadget $G_i$, and let $C_o$ be its last cycle overlapping with gadget $G_o$.
If $C_i$ is disjoint from the cycles of $G_i$ that \T adds to \cmax, then we say that the input signal to the wire is \true; otherwise, if $C_i$ overlaps with one of the cycles of $G_i$ in \cmax, the input value is \false.
If \T does not add $C_o$ to \cmax, then the output signal is \true, and the output signal is \false otherwise.

To ensure that \F always follows the strategy of double-dealing moves, we require that each maximal chain of degree-2 boxes in a \wire gadgets contain at least four boxes.
That way, \F receives at least as many boxes in each chain (and cycle) as \T, and thus for \F being in control is always beneficial~\cite{winning_ways}.

Note that, besides the lower bound on the length of a chain, the size and the embedding of the overlapping cycles in a wire can be chosen freely. %, since all cycles yield 4 boxes for \T independent of their size. 
Thus wires are very flexible in connecting components together, which facilitates the construction. 

\begin{lemma}\label{lem:wire:good}
Let a wire $w$ consist of $2k$ partially overlapping cycles.
Then, under optimal play, if the signal in $w$ changes from \false to \true, then \T can select at most $k-1$ disjoint cycles from $w$ to add to \cmax.
Otherwise, under optimal play, \T can select $k$ disjoint cycles from $w$ to add to \cmax.
\end{lemma}
\begin{proof}
As we show in Lemma~\ref{lem:strategies}, after the first $n$ moves, which \T and \F make in the \variable gadgets, the game enters a loony endgame with \F in control.
If the output signal in the wire matches the input signal, then only one of $C_i$ or $C_o$ of $w$ are in \cmax.
Then \T can select all odd (if $C_i\in\cmax$) or all even (if $C_o\in\cmax$) cycles to add to \cmax, which results in $k$ disjoint cycles.
If the the input signal is \true, and the output signal is \false, then both $C_i$ and $C_o$ are in \cmax.
Then \T can, for example, select $k-1$ odd cycles and $C_o$ to add to \cmax, which again results in $k$ cycles in total.

If, however, the input signal is \false, and the output signal is \true, then neither $C_i$ nor $C_o$ can be in \cmax.
This leaves a chain of $2k-2$ cycles, of which at most $k-1$ disjoint cycles can be selected to be added to \cmax. 
\end{proof}

In our construction we ensure that \T can win only if she gets $k$ disjoint cycles from a wire, and thus under optimal play she cannot flip a signal propagating from a variable from \false to \true.
Flipping a signal from \true to \false is not beneficial for \T, as her goal is to satisfy all the clauses.
Nevertheless, flipping a signal from \true to \false leads to the same number of boxes for her (at least locally within a wire), and is thus allowed.

\subsection{Crossover gadget}\label{sec:crossover}

\begin{figure}[t]
\begin{minipage}[t]{0.48\textwidth}
    \centering
    \includegraphics[page=11]{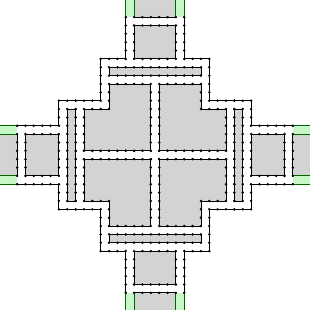}
    \caption{The \crossover gadget. Connections to the adjacent wires are shown in green.}
    \label{fig:crossover}
\end{minipage}
\hfill
\begin{minipage}[t]{0.48\textwidth}
    \centering
    \includegraphics[page=12]{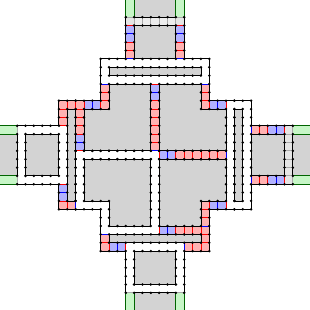}
    \caption{A possible choice of a set of disjoint cycles. The selection has fourth degree rotational symmetry.}
    \label{fig:crossover-disjoint}
\end{minipage}
\end{figure}

Since the graph representing \gpos is not necessarily planar, wires may need to cross each other in our construction.
We describe a \crossover gadget that allows two signals to cross while preserving the signal values.
The gadget has two inputs and two outputs on the opposite sides of the gadget.
Let $C_{1,i}$ and $C_{2,i}$ be the input cycles of the gadget, and $C_{1,o}$ and $C_{2,o}$ be the output cycles (see Figure~\ref{fig:crossover}).
An input cycle $C_{*,i}$ is in \cmax if the corresponding input signal is \true, and otherwise it is \false.
An output cycle $C_{*,o}$ is not in \cmax if the output signal is \true, and otherwise it is \false.

There are four pairwise overlapping cycles $C_a$, $C_b$, $C_c$, and $C_d$ in the middle of the gadget, forming a cross shape.
Only one of these cycles can be added to \cmax.
A choice of which of these cycles is added to \cmax is in one-to-one correspondence to the input signal values (see Figure~\ref{fig:crossover-disjoint}).

\begin{lemma}\label{lem:crossover:good}
Under optimal play, if a signal in a \crossover gadget changes from \false to \true, then \T can select at most $4$ disjoint cycles from the gadget to add to \cmax.
Otherwise, under optimal play, \T can select $5$ disjoint cycles from the gadget.
\end{lemma}
\begin{proof}
If the output signals in the \crossover gadget match the input signals, then only one of each pair $\{C_{1,i}, C_{1,o}\}$ and $\{C_{2,i}, C_{2,o}\}$ are in \cmax.
Since the four center cycles $C_a$, $C_b$, $C_c$, and $C_d$ all share a single square, only one of these four cycles can be chosen.
Then \T can select a corresponding cycle from the middle of the gadget, and two more cycles from each signal.
For example, a selection of five disjoint cycles for the case when the first input signal is \false and the second is \true is shown in Figure~\ref{fig:crossover-disjoint}.
If an input signal is \true, and the corresponding output signal is \false, then both $C_{*,i}$ and $C_{*,o}$ are in \cmax.
Then \T can, for example, make exactly the same choice as in the case where the output signal would have been \true.

Assume now, w.l.o.g., that the signal corresponding to $C_{1,i}$ and $C_{1,o}$ changes from \false to \true in the gadget.
That is, neither $C_{1,i}$ nor $C_{1,o}$ are in \cmax.
Let $C'$ and $C''$ be the cycles in the gadget adjacent to $C_{1,i}$ and $C_{1,o}$ respectively.
Thus, among cycles $C'$, $C''$, $C_a$, $C_b$, $C_c$, and $C_d$ at most two cycles can be in \cmax, and therefore at most four cycles can be chosen to be in \cmax.
\end{proof}

\subsection{Or gadget}\label{sec:or}

\begin{figure}[b]
    \centering
    \includegraphics[page=24]{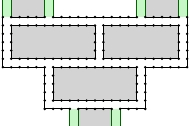}
    \hfill
    \includegraphics[page=25]{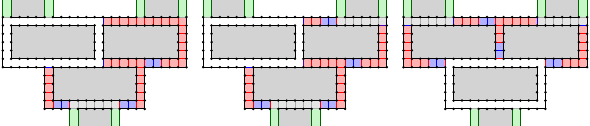}
    \caption{The \org gadget (left) and the three possible combinations of the input values. From left to right: two \true inputs, one \true and one \false input, and two \false inputs. The boxes highlighted in grey belong to a cycle in an adjacent \wire gadget.}
    \label{fig:or-gadget}
\end{figure}

The \org gadget consists of three pairwise overlapping cycles (see Figure~\ref{fig:or-gadget} (left)).
Two of the cycles partially overlap with an end cycle of an input wire, and one cycle partially overlaps with the output cycle.
Let $C_{1,w}$ and $C_{2,w}$ be the last cycles of the two input \wire gadgets, and let $C_{1,i}$ and $C_{2,i}$ be the cycles of an \org gadget adjacent to these two wires respectively.
Let $C_o$ be the third cycle of the \org gadget, which is adjacent to an output wire.
Cycles $C_{1,w}$ and $C_{2,w}$ are not in \cmax if the input from their corresponding wire is \true, and are in \cmax if their input is \false.
If $C_o$ is not in \cmax then the output of the \org gadget is \true, and if it is in \cmax then the output value is \false.
Only one of the three cycles in the \org gadget can be selected to be added to \cmax, and thus the output of the gadget can be \true only if one of $C_{1,w}$ or $C_{2,w}$ is in \cmax.

\begin{lemma}\label{lem:org:good}
Under optimal play, if both input signals in an \org gadget are \false but the output signal is \true, then \T cannot add a single cycle from the gadget to \cmax.
Otherwise, under optimal play, \T can select $1$ cycle from the gadget to add to \cmax.
\end{lemma}
\begin{proof}
First consider the case when one of the input signals in the \org gadget is \true.
W.l.o.g., let the signal from the first wire be \true, that is $C_{1,w}$ is not in \cmax.
Then \T can select $C_{1,i}$ to add to \cmax and thus the output from the \org gadget would correspond to \true.
\T may as well choose $C_o$ to add to \cmax and make the output of the gadget to be \false.
In either case, one cycle from the gadget is in \cmax.

If both input signals are \false, then both cycles $C_{1,w}$ and $C_{2,w}$ are in \cmax.
Thus none of $C_{1,i}$ and $C_{2,i}$ can be in \cmax.
If at the same time the output of the \org gadget is \true, then $C_o$ is not in \cmax, and thus \T cannot select a single cycle to add to \cmax from this \org gadget.
\end{proof}

\subsection{Variable gadget}\label{sec:variable}

\begin{figure}[b!]
\begin{minipage}[t]{0.46\textwidth}
    \centering
    \includegraphics[page=19]{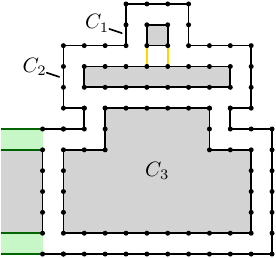}
    \caption{The value-setting component of the \variable gadget. There are two non-loony moves (yellow) available, of which only one can be played as a non-loony move.}
    \label{fig:variable}
\end{minipage}
\hfill
\begin{minipage}[t]{0.5\textwidth}
    \centering
    \includegraphics[page=22]{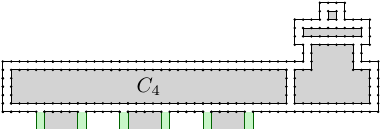}
    \caption{The complete \variable gadget consisting of the value-setting component and the fan-out component. Outgoing wires are shown in green.}
    \label{fig:var-example}
\end{minipage}
\end{figure}

\begin{figure}[t]
    \centering
    \includegraphics[page=20]{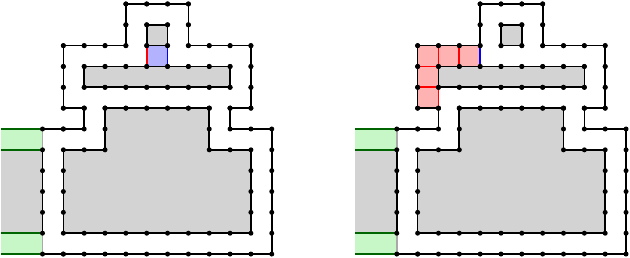}
    \caption{The variable is set to \false (left) and \true (right).}
    \label{fig:variable-set}
\end{figure}

The \variable gadget is responsible for the assignment of \true and \false values to the variables of the \gpos instance.
It consists of two components: the \emph{value-setting component} (see Figure~\ref{fig:variable}) designed to set the value of the variable, and the \emph{fan-out component} designed to duplicate the variable signal.
The whole construction is presented in Figure~\ref{fig:var-example}.
Let $C_1$, $C_2$, and $C_3$ be the three cycles in the value-setting component.
The \variable gadget is the only gadget that contains non-loony moves; there are two non-loony moves (shown in yellow in the figure) at the intersection of $C_1$ and $C_2$.

As we show later, optimal play by both \T and \F is to set all the variables in the first $n$ moves, such that \F always sets a variable to \false and \T---to \true.
Figure~\ref{fig:variable-set} shows the two possible value assignments of the variable gadgets.
To set a variable to \false, \F plays one of the non-loony moves in the corresponding \variable gadget.
Then \T responds by claiming the one box available (see Figure~\ref{fig:variable-set} (left)).
This results in the cycles $C_1$ and $C_2$ getting merged.
To set a variable to \true, \T opens a side chain of $C_2$ (see Figure~\ref{fig:variable-set} (right)).
Then \F responds by claiming every box in the opened chain, and proceeds to setting the next variable.
Note that after \T's move the non-loony moves in the gadget become loony moves (as they are now a part of a long chain).

At this point we make two observations which will be useful when proving correctness of the construction and the properties of the optimal play in Section~\ref{sec:construction}.
First, observe that the non-loony moves come in pairs, one in each variable, such that, for each pair, either both moves in the pair are still non-loony or neither is anymore.
We refer to them as \emph{non-loony pairs}.
Second, note that in the process of assigning values to the \variable gadgets, \T gets a box for each variable set to \false by \F, and zero boxes for each variable set to \true by herself.

Once the value of a variable is set, it propagates to the outgoing wires through the fan-out component of the \variable gadget.
The fan-out component simply consists of one cycle $C_4$ overlapping with the cycle $C_3$ (see Figure~\ref{fig:var-example}), to which multiple wires can be attached.
After the variable is set, \T can add at most two cycles from it to \cmax.
Then, if the variable is set to \false, cycle $C_4$ has to be one of the two selected cycles, and thus the signal propagated into the wires is \false.
If the variable is set to \true, \T can add $C_1$ and $C_3$ to \cmax, and thus propagate the \true value into the wires.

\begin{lemma}\label{lem:var:good}
Under optimal play, after a \variable gadget is assigned a value, if it is set to \false but the output signal is \true, then \T can add at most $1$ cycle from the gadget to \cmax.
Otherwise, under optimal play, \T can add $2$ cycles from the gadget to \cmax.
%If a \variable gadget has no non-loony moves left and the \variable gadget is chosen to be \false, this signal value will be propagated to the output wires under optimal play.
\end{lemma}
\begin{proof}
As we show in Lemma~\ref{lem:strategies}, optimal play of both \T and \F results in them setting all the variables according to the rules described above in the first $n$ moves.
Afterwards the game enters a loony endgame with \F in control.

If a \variable gadget is set to \true, then there are three cycles left in the gadget: two overlapping cycles $C_3$ and $C_4$, and the cycle $C_1$ connected to $C_3$ by a chain.
Then \T can select $C_1$ and one of $C_3$ or $C_4$ to add to \cmax.

If the \variable gadget is set to \false, then there are still three cycles left in the gadget, but now these cycles are forming a chain where each consecutive pair of cycles is overlapping.
Now, if the output value is \true then $C_4$ cannot be in \cmax, and from the remaining two cycles, only one can be selected to be added to \cmax.
\end{proof}

\subsection{Clause gadget}

\begin{figure}[t]
    \centering
    \includegraphics[page=47]{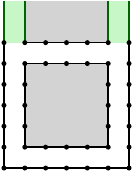}
    \hfil
    \includegraphics[page=48]{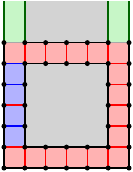}
    \hfil
    \includegraphics[page=49]{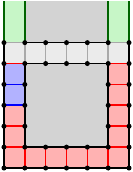}
    \caption{The \clause gadget (left) yields four boxes to \T if the input signal is \true (middle), and only two boxes when the input is \false (right). Boxes highlighted in grey belong to the last cycle in the adjacent wire.}
    \label{fig:clause}
\end{figure}

Finally, we describe a \clause gadget that yields more boxes to \T if the signal entering the clause corresponds to \true.
A \clause gadget is simply an extra cycle extending the end of a \wire gadget to an odd length.
Figure~\ref{fig:clause} shows the gadget, and the two possible assignments of this gadget.
Whenever the signal is \true, it is possible for \T to create a disjoint cycle in the gadget which gives her four boxes.
If the signal is \false, \T can only make a chain in this gadget which yields only two boxes.

\begin{lemma}\label{lem:clause}
Under optimal play, the clause gadget yields at most $4$ boxes to \T if the input signal is \true, and at most $2$ boxes if the input signal is \false. 
\end{lemma}
\begin{proof}
If the input signal to the clause gadget is \true, the adjacent cycle to the clause gadget is not in \cmax.
Therefore, a the cycle of the gadget can be added to \cmax.
When in the loony endgame, this cycle yields four boxes to \T after \F makes a double-dealing move.

Otherwise, if the input signal is \false, the adjacent cycle is in \cmax, and from the \clause gadget only a chain is left.
This chain yields only two boxes to \T after \F makes a double-dealing move.
\end{proof}

\section{Players' strategies and PSPACE-completeness}\label{sec:construction}
With the gadgets described above, we construct a \DB instance for any \gpos instance such that \T can win the \DB instance if and only if she can win the corresponding \gpos instance.
We lay out the \variable gadgets, attach a corresponding number of \wire gadgets, pass the wires through \org gadgets, using \crossover gadgets to cross signals, and finally connect wires to the \clause gadgets.
An example of our construction is given in Figure~\ref{fig:example-game} in the appendix.
%shows an example \DB instance that is created when reducing from the \gpos formula $(w \vee x) \wedge (w \vee y) \wedge (x \vee z)$. This formula was also used in the graph representation shown %in Figure~\ref{fig:graph_pos}.

The initial score we set to the \DB instance depends on the number of gadgets of each type in the construction.
By Lemma~\ref{lem:max_cycles} the total score in the loony endgame depends on the number of disjoint cycles $c$, the number of boxes $k$ with degree higher than $2$, and their total degree $T$.
The configuration of the loony endgame, and thus the values $k$ and $T$, is changed only when the \variable gadgets are being assigned their values.
We will argue below, that under optimal play, exactly half of the variables are set to \true and half are set to \false.
Thus the total values of $k$ and $T$ are the same, no matter which variables are assigned to which values.
If \T can satisfy $\mathcal{F}$, by Lemma~\ref{lem:max_cycles}, she can claim $4c+T-2k-4$ boxes in the loony endgame, and $n/2$ boxes from the variables set to \false.
Let $N$ be the total number of unclaimed boxes in our \DB instance.
Then, \F gets $N-n/2-(4c+T-2k-4)$ boxes.
We set the initial scores of \T and \F such that \T's final score is one larger than \F's if she can satisfy $\mathcal{F}$.
Otherwise, her score will be strictly less than \F's.

Next, we describe the optimal strategies for \T and \F, both before the loony endgame is entered and in the loony endgame.

\subparagraph{Optimal strategies for \T and \F in the loony endgame}
We start by summarizing both strategies in the loony endgame, assuming that all variables have already been assigned a value using the moves we have described in Section~\ref{sec:variable}.
As we argue below, \F can always ensure that he is in control of the loony endgame.
It is always beneficial for \F to stay in control, as all the chains and cycles in the loony endgame configuration yield at least as many boxes to him than to \T.

In the loony endgame, \T can choose which chains and cycles to open.
To maximize her score, \T is going to select a maximum number of disjoint cycles \cmax in the loony endgame (see Lemma~\ref{lem:max_cycles}).
This can be done by first making a loony move in all chains, to which \F responds by claiming all but two boxes, finishing with a double-dealing move in order to stay in control.
Afterwards, \T makes loony moves in the remaining cycles, to which \F responds again by claiming all but four boxes, finishing with a double-dealing moves each time, except for in the final cycle.

\subparagraph{Optimal strategy for \T before the loony endgame}
\T's strategy before the loony endgame is to set enough \variable gadgets to \true in order to satisfy all the clauses.
By Lemmas~\ref{lem:max_cycles} and~\ref{lem:clause}, \T gains more boxes from each satisfied clause.
Therefore, the optimal strategy for \T is to claim the boxes opened by \F when setting variables to \false, and to set variables to \true, by using a loony move in a side chain of cycle $C_2$ of the variables.

As we show in Lemma~\ref{lem:strategies}, if \F deviates from setting variables to \false, and plays a loony move when there are non-loony moves available, \T can adopt \F's strategy and dominate the rest of the game by ensuring that she ends up in control when the loony endgame is entered.

\subparagraph{Optimal strategy for \F before the loony endgame}
\F's strategy is to ensure that he is in control when the loony endgame starts, and it can be described completely as responses to what \T does.
By our assumption the number of variables in $\mathcal{F}$ is even, thus initially the number of non-loony move pairs is even.
\F's strategy is then to keep the number of non-loony move pairs even at the start of every \T's turn.
Then, once the number of non-loony moves reaches zero (and the loony endgame is reached), it is \T's turn, and \F is in control.
Specifically, \F responds to \T's moves in the following way:
\begin{itemize}
    \item If \T follows optimal play and makes a loony move in a variable to set it to \true, then \F simply claims all boxes in the chain opened by \T (without making a double-dealing move), and makes a non-loony move in another variable to set it to \false.
    \item If \T deviates from her strategy by making a non-loony move, setting a variable to \false, there must be at least one other non-loony move pair available to \F.
    Therefore, \F claims the boxes opened by \T, and makes a non-loony move, thereby setting another variable to \false.
    The number of non-loony pairs is again even at the start of \T's next turn.
    \item If \T deviates from her strategy by opening a chain with a loony move that does not remove a non-loony pair, \F responds with claiming all but two (or four in case of a cycle) boxes and ends with a double-dealing move.
    The number of non-loony pairs remains even before \T's next turn.
\end{itemize}
Using this strategy, \F can set a variable to \false each time \T sets a variable to any value, as well as gain control in the loony endgame. 

\medskip

Note that the order of moves in these strategies is not enforced.
\T can play loony moves she would play in the loony endgame even if there are still non-loony moves available, as long as these moves do not interfere with the values set (or to be set) in the corresponding variables.
For \F it is optimal to simply respond to these moves as if the game was already in the loony endgame, since otherwise he would be in danger of losing control.
Indeed, if \F does not make a double-dealing move, the number of non-loony moves will no longer be even at the start of \T's turn, and \F loses control of the loony endgame.
Thus, it is not more beneficial for any player to make a move in any other gadget than the \variable gadgets while there are still variables that have not been set.

\begin{lemma}\label{lem:strategies}
Deviating from the strategies described above is sub-optimal for \F and cannot be more beneficial for \T.
\end{lemma}
\begin{proof}
Trivially, \T and \F always claim open boxes before making their move, except when \F makes double-dealing moves.
Otherwise the opponent can claim these boxes in their next move.

First, consider the strategies in the loony endgame.
If \T deviates from her strategy and does not select the maximum number of disjoint cycles, by Lemma~\ref{lem:max_cycles} her score will be too low and she loses the game.
Therefore, the loony endgame strategy for \T as described above is optimal.

If, at any point in the loony endgame, except for his last move, \F does not make a double-dealing move, he loses control.
Since being in control is always beneficial in our construction, this play is sub-optimal.

The strategies described for before the loony endgame are also optimal.
Observe that, under the described strategies, the value-setting component of a variable yields the same number of boxes to \T independent whether it is set to \true or to \false.
Indeed, if it is set to \true, the component contains three boxes with degree $3$, while setting the variable to \true does not give any boxes to \T; if the variable is set to \false, the component contains two boxes with degree $3$, but setting the value gives \T one box.
Thus, the value-setting component contributes the same number of points to \T's final score independent of the value.

If \T deviates from her strategy by making a non-loony move and setting a variable to \false, she loses one box to \F.
Furthermore, setting a variable to \false can never help \T to satisfy formula $\mathcal{F}$.
Thus, such a move is sub-optimal.

If \T deviates from her strategy by making a loony move in any other gadget than the variable gadget, there are two options: either she makes a move that leads to the same score as the strategy described above, or she makes a move that contradicts the setting of the variables and reduces her total score.
The former case does not have any bad repercussions for \T. 
\F will respond with a double-dealing move, otherwise \T would take control of the endgame.
Thus, we can reorder the sequence of \T's moves and assume that she first sets all the variables.
However, in the latter case, the move reduces the number of possible disjoint cycles, and thus leads to \T's loss in the game.
Therefore, deviating from the strategy above is never more beneficial for \T.

If \F deviates from his strategy before the loony endgame, then \T can adopt his strategy and ensure that the number of non-loony move pairs is even at the start of each of \F's turn.
Since, if \F is not in control of the loony endgame, he loses the game, deviating from his strategy is not optimal.
\end{proof}

\begin{theorem}
\DB is PSPACE-complete.
\end{theorem}

\begin{proof}
A game of \DB is finished after a polynomial number of turns.
Thus, all possible sequences of moves can be explored using polynomial sized memory.
This implies that \DB is in PSPACE.

We now show that \DB is PSPACE-hard.
Given a \gpos formula $\mathcal{F}$, we construct a \DB instance $\delta$ following the description above.
We argue that \T can win $\mathcal{F}$ if and only if \T can win $\delta$.

If \T can win $\mathcal{F}$, then there must be a variable assignment following the \gpos rules such that every clause is connected to at least one variable which has been set to \true.
Therefore, there can be at most $n/2$ variables that need to be set to \true by \T.
Hence, \T can set the corresponding \variable gadgets in $\delta$ to \true, and if needed set the remaining variables available to her to \true in any order.
Thus, by Lemmas~\ref{lem:wire:good}--\ref{lem:clause}, \T can propagate the \true values down to all the clauses, that is, she can select the maximum number of disjoint cycles from all the gadgets, including all the clause gadgets, leading to the winning score in $\delta$.

In order for \T to win $\delta$, the set of disjoint cycles \cmax that she selects must contain a cycle from every \clause gadget, and the maximum number of cycles from all the other gadgets.
By Lemmas~\ref{lem:wire:good}--\ref{lem:clause}, this can be done only if the output signals from each gadget conform to their input signals, and thus there must be a set of \variable gadgets set to \true whose signal is propagated all the way down to all the \clause gadgets. 
%\T can only gain the extra boxes from the \clause gadgets if and only if they are connected to a \variable gadget which has been set to $\true$. 
In $\delta$ \T and \F have to alternate choosing which \variable gadgets get set to \true and \false, respectively. 
%Thus, if \T wins there must be a way of assigning variables such that every \clause gadget obtains a \true signal.
This assignment can be used as a winning strategy for \T to win \gpos game on $\mathcal{F}$.
%Hence, if \T can win $\delta$ then \T can also win $\mathcal{F}$

Thus, \DB is PSPACE-complete.
\end{proof}

\section{Conclusion}
In this paper we proved \DB to be PSPACE-complete, resolving a long-standing open problem. 

There exist a number of other intriguing open problems related to \DB. Does restricting the game to a $k \times n$ grid for a small $k$ make the game easier? How large does $k$ need to be to make the problem PSPACE-hard or even just NP-hard? These are challenging questions, given that even for a $1 \times n$ grid \DB is not yet fully understood~\cite{collette2015narrow,unsolved2002,JobsonSWW17}.

 Another direction of further research is the computational complexity of variants of \DB, in particular misère \DB~\cite{collette2015narrow}, of \DB on other grids or even of variants of \DB with more than two players as it was originally described by Lucas~\cite{lucas1883recreations}.  One variant that our result resolves is \emph{Dots \& Polygons}, since the reduction from \DB to \emph{Dots \& Polygons} that was used to prove NP-hardness~\cite{buchin_et_al:LIPIcs:2020:12237} now directly also shows PSPACE-hardness.

Our result can be interpreted as proving that \emph{Strings and coins} restricted to grid graphs is PSPACE-complete. What is the complexity of \emph{Strings and coins} on other restricted graph classes, for instance outerplanar graphs (which generalize $1 \times n$ grids)?

This may also be a good moment to revisit other games, which are known to be PSPACE-complete on general graphs, but for which the complexity on grid graphs is open. 
This, for instance, includes \emph{NoGo}, \emph{Fjords} (on hexagonal grids), \emph{Cats-and-Dogs} and \emph{GraphDistance}, which are known to be PSPACE-complete for planar graphs~\cite{BurkeH19,BurkeHHH16}. 

\newpage
\bibliography{References}

\appendix

\newpage

\section{Example game}
\begin{figure}[h]
    \centering
    \includegraphics[page=44]{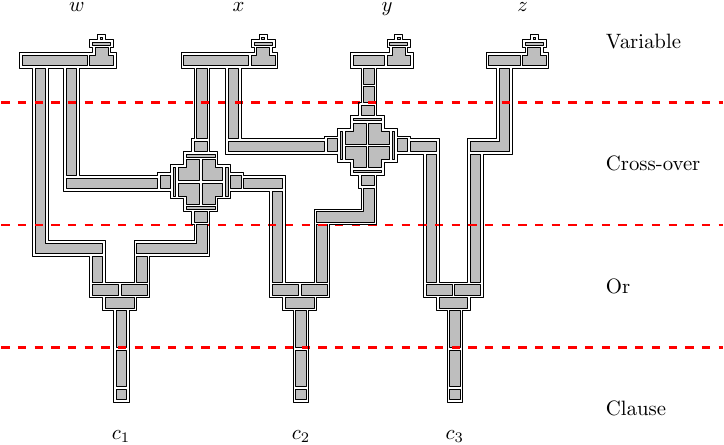}
    \caption{Example reduction from the \gpos formula $(w \vee x) \wedge (w \vee y) \wedge (x \vee z)$. The construction can be divided into four sections: a \variable, \crossover, \org, and \clause section. Each section contains only the corresponding gadgets and \wire gadgets that connect different gadgets together.}
    \label{fig:example-game}
\end{figure}

\clearpage
\section{Omitted proofs}

\lemmaxcycles*
\begin{proof}
Let Fred be in control of the game.
To simplify the argument, w.l.o.g., we assume that the last move made by \T is made in a cycle.
Let $c$ denote the number loony moves made by \T in a disjoint cycle and let $\ell$ be the number of loony moves made by \T in chains. All but the last loony move in a disjoint cycle or chain yield 4 or 2 boxes for \T, respectively. Thus, the score gained by \T in the loony endgame is \[4c + 2\ell - 4.\]
Consider the dual graph $G=(V,E)$ to the \DB instance.
In it, a node corresponds to a box, and an edge connects two nodes if the two corresponding adjacent boxes do not have a line drawn between them.
Suppose $G$ has $k$ nodes with degree higher than $2$.
We define $T$ to be the sum of the degrees of these nodes:
\[T = \sum_{\{v \in V | \mathit{degree}(v) > 2\}} \mathit{degree}(v).\]
A loony move on a disjoint cycle does not change $T$, since all disjoint cycles only contain boxes of degree 2.
A loony move on a chain, however, decreases the degree of the box at both ends of the chain by 1.
Furthermore, whenever the degree of a box reduces from 3 to 2 the degree of this box is no longer counted in $T$.
Thus \[T = 2\ell + 2k,\] which means the score for \T will be \[4c + T - 2k - 4.\]
Since $T$ and $k$ are fixed, the score is maximized when the number of loony moves in disjoint cycles is maximized.
\end{proof}
\end{document}